\documentclass[conference,twocolumn]{IEEEtran}
\ifCLASSINFOpdf
\else
\fi
\usepackage{amsthm,amssymb}
\usepackage{epstopdf}
\usepackage{mathtools}
\usepackage{hyperref}

\newcommand{\X}{\mathbf{X}} 
\newcommand{\Y}{\mathbf{Y}}
\newcommand{\A}{\mathbf{A}}
\newcommand{\R}{\mathbf{R}}
\newcommand{\B}{\mathbf{B}}
\newcommand{\C}{\mathbf{C}}
\newcommand{\F}{\mathcal{F}}
\newcommand{\SNR}{\mbox{SNR}}
\newcommand{\I}{\mathbb{I}}

\def\spn{\mathop{\mathrm{span}}}
\def\diag{\mathop{\mathrm{diag}}}
\newcommand*{\tm}[1]{{\scriptstyle[#1]}}

\newtheorem{theorem}{Theorem}[section]

\begin{document}
%
\title{Achievable Degrees of Freedom in MIMO Correlatively Changing Fading Channels}


\author{
  \IEEEauthorblockN{Mina Karzand}
  \IEEEauthorblockA{Massachusetts Institute of Technology\\
    Cambridge, USA\\
    Email: mkarzand@mit.edu} 
  \and
  \IEEEauthorblockN{Lizhong Zheng}
  \IEEEauthorblockA{Massachusetts Institute of Technology\\
    Cambridge, USA\\
    Email: lizhong@mit.edu}
}


%


\maketitle

\begin{abstract}
The relationship between the transmitted signal and the noiseless received signals in correlatively changing fading channels is modeled as a nonlinear mapping over manifolds of different dimensions. Dimension counting argument claims that the dimensionality of the neighborhood in which this mapping is bijective with probability one is achievable as the degrees of freedom of the system. We call the degrees of freedom achieved by the nonlinear decoding methods the \emph{nonlinear degrees of freedom}.
\end{abstract}



%
\IEEEpeerreviewmaketitle

\section{Introduction}
The noncoherent systems in which neither the transmitter nor the receiver know the fading coefficients of the wireless channel are of great importance. 

The dependency of capacity of fading channels with additive white gaussian noise on the $\SNR$ is a critical measure in designing the communication systems. In low $\SNR$ regime, the capacity generally increases linearly with $\SNR$; whereas in high $\SNR$ regime, the dependency of the capacity to $\SNR$ is logarithmically. The degrees of freedom (DOF) of channel is defined as  the pre-log factor of the capacity in high SNR regime.

The classic approach towards the achievability of DOF in noncoherent systems exploits a set of training signals in a block of time to estimate all the unknown parameters of the channel. In the transmission phase of communications, the message is recovered at the receiver using the estimated fading coefficients.

In this paper, the dimension counting argument in an achievability scheme of degrees of freedom of fading channels is formalized. The geometric intuition behind dimension counting argument counts the dimensionality of the subspace in the output signal space which is reachable only by changing the transmitted signal. 

Theorem \ref{thm:DOF} formalizes this intuition. In this theorem, we consider a mapping from a $D$ dimensional neighborhood in the input signal space to the noiseless output signal space. It is proved that if this mapping is bijective with probability one in this neighborhood, then $D$ degrees of freedom is achievable in this system.  This approach is used to propose an achievability scheme for degrees of freedom of correlatively changing channels. 

We consider the multiple-input multiple-output (MIMO) systems with $n_t$ transmit antennas and $n_r$ receive antennas. In correlatively changing channels, the fading coefficients are time varying. But this variation in time is not stochastically independent. In a block of length $T$, there are only $Q<T$ statistically independent elements which determine the fading coefficients in this block. This implies a low rank correlation matrix over the fading coefficients. Whitening process of the fading coefficients uses correlation matrix $\mathbf{K}_{\mathbf{H}}=\A^{\dagger}\A$. Matrix $\A$ of size $Q \times T$ gives the linear equations determining the fading coefficients in a block of length $T$ from the $Q$ independent elements. This matrix plays an important role in the achievability conditions of our proposed scheme. This channel model was first introduced and discussed in \cite{Liang}.

Having the classical training approach in mind, one might try to estimate the unknown parameters of the channel in each block and then communicate the message knowing the fading coefficients. This is the approach taken in \cite{Liang} where it is proved that this is the optimal strategy in terms of the achievable DOF in single-input single-output (SISO) systems. In a block of length $T$, $Q$ symbols are assigned to gather information about the fading coefficients in the training phase and $T-Q$ symbols are used to convey the message in the transmission phase. Thus there are  $(1-Q/T)$ degrees of freedom per transmitted symbol.

In the same paper, there is a conjecture about the MIMO systems which states that if $n_{t} \leq \min\{{n_{r}},T/2\}$, the pre-log factor of the system is $n_{t}(1-n_{t}Q/T)$. In a block of length $T$, there are $n_{t}^2 Q$ independent unknown elements which describe the fading coefficients in this block. Thus, $n_{t}^2Q$ symbols are assigned to gather information about the fading coefficients and $n_{t} T-n_{t}^2Q$ symbols are used to transmit information. The loss in the number of DOF due to the training is $n_t^2 Q$ which is proportional to the rank of the correlation matrix in this case. This conjecture is proved to be wrong for single-input multiple output (SIMO) systems in \cite{Riegler}. We prove in \cite{Karzand} that this conjecture is not true in general for MIMO systems either and strictly higher number of DOF can be achieved using nonlinear decoding algorithms.

In \cite{Riegler} and \cite{Riegler2}, SIMO systems are studied. Hironaka's theorem on resolution of singularities in algebraic geometry is used to prove that the pre-log factor of $(1-1/T)$ is achievable as long as $T>2Q-1$ under some constraints over the correlation matrix of the fading coefficients. It is proved that the lower and upper bounds for the DOF of correlatively changing fading SIMO channels match. The number of DOF for these systems is $\min\left[1-1/T,n_r(1-Q/T)\right]$.

In \cite{Riegler3} the MIMO channels under this statistics is studied and a lower bound on DOF is presented.

 In this paper, we present an achievability scheme for MIMO systems. This scheme uses dimension counting argument as presented in theorem \ref{thm:DOF}. We model the transformation from the transmitted signals to the noiseless received signals as  a nonlinear mapping over the manifolds of different dimensions. This mapping is studied and proved to be bijective with probability one over a neighborhood in input signal space. The dimensionality of this space can be achieved as number of degrees of freedom.  

If the transmitted signals span an $M$ dimensional linear subspace in $\mathcal{C}^T$, the received signals live in an $MQ$ dimensional linear subspace denoted by $\F_{\A}(\X)$. 
In the regime where $n_r \geq MQ$ and $T \geq M(Q+1)$, this mapping is bijective with probability one over the input signal space defined as all the $M$ dimensional linear subspaces in $\mathcal{C}^T$. We define $M^*=\min\left[n_t,\lfloor{\frac{n_r}{Q}}\rfloor,\lfloor\frac{T}{Q+1}\rfloor\right].$ Thus,  $M^*(T-M^*)$ degrees of freedom are achievable in this system. 

\section{Dimension Counting Argument}



To understand the geometric interpretation of dimension counting argument, the local behavior of the signals should be studied. Assume that the transmitted signal in a block of time is a specific realization of input signal in the input signal space, $\X_0 \in \mathcal{X}$. Also assume that channel operates based on a specific realization of the fading coefficient $\mathbf{H}_0$ and noise $\mathbf{W}_0$. Clearly, the underline probability distribution which generates the fading coefficients $\mathbf{H}_0$ in a block of time satisfies the dependency constraints imposed by the given channel model. These realizations of transmitted signal and the channel parameters imply a noisy received signal living in the output signal space, $\Y_{0,\mbox{noisy}}\in \mathcal{Y}$. 

Obviously, having fixed the transmitted signal and the realization of the fading coefficients of the channel as the realization of the additive noise changes, the noisy received signal can move along all the possible dimensions of the output signal space. This local behavior of the received signal is not of interest in dimension counting argument. 

Define $\Y$ as the noiseless received signal. The initial realization of $\X_0$ and $\mathbf{H}_0$ implies the noiseless received signal $\Y_0$. Altering the realization of $\mathbf{H}$ and  $\X$ from $\mathbf{H}_0$ and $\X_0$ locally, we can move $\Y$ in a neighborhood around $\Y_{0}$. This neighborhood around $\Y_{0}$ is a subset of the output signal space $\mathcal{Y}$ which was reachable by noisy received signal.

To find the number of DOF, we categorize the dimensions of the neighborhood reachable in \emph{noiseless received signal} into three categories:

\begin{enumerate}
		\item The dimensions of the neighborhood around $\Y_0$ that is reachable only by altering transmit signal $\X$ from $\X_0$.
	\item The dimensions of the neighborhood around $\Y_0$ that is reachable only by altering the realization of $\mathbf{H}$ from $\mathbf{H}_0$. 
	\item The dimensions of the neighborhood around $\Y_0$ that is reachable by altering transmit signal $\X$ from $\X_0$ or changing $\mathbf{H}$ from $\mathbf{H}_0$. 
\end{enumerate}

Theorem~\ref{thm:DOF} describes how in the high SNR regime communications can take place along the dimensions in the first category. If the dimensionality of the space in the first category is $D$, this theorem proposes the coding scheme achieving $D$ degrees of freedom.

\begin{theorem}[Achievability of Degrees of Freedom]
	\label{thm:DOF}
Let's model the communication system as a continuous mapping which transforms the transmitted signal at the input signal space to the noiseless received signal at the output signal space. Assume that there exists a $D$ dimensional neighborhood in the input signal space which is mapped to a $D$ dimensional neighborhood in the output signal space. If the mapping from transmitted signals to the noiseless received signal is bijective with measure one with respect to Lebesgue measure over this neighborhood, then the DOF of $D$ is achievable in this communication system.
\end{theorem}

\begin{proof}
	To prove the achievability of $D$ degrees of freedom in the communication system with described properties, we propose a coding scheme achieving the rate $R \approx D \log{\mbox{SNR}} + o(\log{\mbox{SNR}})$.

To achieve the $D$ degrees of freedom in this neighborhood, QAM modulation is performed in each of the $D$ dimensions of input space which is conserved in output space. Define $d_{\text{min,x}}$ to be the  minimum distance of the codewords in input space. In the communication channel with signal to noise ratio $SNR$, we can assume the input power constraint implies $\mathbb{E}[\| x \|^2] \leq 1$ and stationary noise has power spectral density $1/\text{SNR}$.

	Define $d_{\text{min,y}}$ as the minimum distance between noiseless received codewords in the output space. Since the mapping is one-to-one with probability one in this space, the eigenvalues of the Jacobian of this mapping is strictly positive with probability one. Using the continuity of measure lemma we can prove that for any $\epsilon>0$, there exists $\sigma_0>0$ such that the minimum eigenvalue of the Jacobian of the mapping is greater than $\sigma_0$ with probability $1-\epsilon$. Thus, with probability $1-\epsilon$ we would have $d_{\text{min,y}}\geq d_{\text{min,x}} \sigma_0$. 

	In the fading channel with the Rayleigh fading coefficients and noise power density $1/\text{SNR}$, the probability of error vanishes as long as $d_{\text{min,y}} \gg 1/{\sqrt{\text{SNR}}}$. Thus, the probability of error vanishes as long as $d_{\text{min,x}} \sigma_0 \gg 1/{\sqrt{\text{SNR}}}$. Setting $d_{\text{min,x}}=\frac{1}{\sigma_0 \text{SNR}^{1/2-\delta}}$ for any $\delta>0$, the probability of error vanishes. The power constraint implies that in each dimension, QAM would give $ (2/d_{\text{min,x}})^2$ codewords. Thus, with probability $1-\epsilon $ the total number of codewords would be $(2/d_{\text{min,x}})^{2D}=(2 \sigma_0 \text{SNR}^{1/2-\delta})^{2D}$ and the achievable rate is $(1-\epsilon)(1-2\delta)D\log(SNR)+o(\log SNR)$.
\end{proof}

In this paper, to prove the achievability of degrees of freedom in correlatively changing fading channels the following scenario is proposed: a $D$ dimensional subspace of the input signal space is chosen. The mapping from this input subspace to the noiseless received signal in the output subspace is studied. If this mapping is bijective with probability one, $D$ degrees of freedom is achievable. To prove the bijective property of the mapping, a decoding algorithm is proposed which recovers the transmitted signal from the received noiseless signal with probability one. 

Note that constraining the transmitted signal to live in the $D$ dimensional subspace corresponds to transmission of training signal. One example of the training scheme which makes the decodability possible in each case is explained. Also note that the decodability is possible under generic parameters of the channel. There are non-generic cases in which the proposed decoding algorithm is not successful with probability one. Clearly, this does not mean that the $D$ degrees of freedom are not achievable. But since we are studying the achievability of the DOF, we are going to give the generic conditions under which the proposed decoding algorithm corresponding to the proposed training scheme succeeds with probability one. These conditions would be interpreted as a set of sufficient but not necessary conditions for achievability of $D$ DOF. These generic situations are described as the a set of conditions over the parameters of the channel. We call them the recovery conditions. 

\section{Notations and Channel Model}

We study the multiple antenna systems with $n_t$ transmit antennas, $n_r$ receive antennas, Rayleigh fading coefficients between pairs of transmitters and receiver, additive white gaussian noise and received signal to noise ratio $\text{SNR}$.

The transmitted signal from $m$th transmit antenna at time $t$ is denoted by $x_{m}\tm{t}\in \mathcal{C}$. The matrix $\X \in \mathcal{C}^{n_t \times T}$ contains all the transmitted signals in a block of length $T$. Its $m$th row, $\mathbf{x}_m\in \mathcal{C}^{1 \times T}$, is the transmitted signal from the $m$th antenna in a block of length $T$. The $t$th column of matrix $\X$, $\mathbf{x}\tm{t}\in \mathcal{C}^{n_t \times 1}$, is the transmitted signal from all $n_t$ antennas at time $t$. Similarly, $y_{n}\tm{t}\in \mathcal{C}$ is the noiseless received signal from $n$th antenna at time $t$ and matrix $\Y\in \mathcal{C}^{n_r \times T}$ contains the noiseless received signals from $n_r$ antennas at a block of length $T$. Its $n$th row and $t$th column are denoted by $\mathbf{y}_n$ and $\mathbf{y}\tm{t}$. The noisy received signal from $n$th antenna at time $t$ is denoted by $y_{\text{noisy},n}\tm{t}$.

In the time varying model of the MIMO fading channel, the matrix $\mathbf{H}\tm{t}\in \mathcal{C}^{n_r\times n_t}$ contains the fading coefficients which determine the relationship between the transmitted signal and the  noiseless received signal at time $t$:

\[\mathbf{y}\tm{t}= \mathbf{H}\tm{t} \mathbf{x}\tm{t}.\]

The fading coefficient between $m$th transmit antenna and $n$th receive antenna, $h_{m,n}\tm{t}$, has normal complex gaussian distribution for $m=1,\cdots, n_{t} \quad  \mbox{and} \quad n=1,\cdots,n_{r}.$ The fading coefficients between different pairs of transmitters and receivers are independent of each other.


The noisy received signal at the receiver is $\Y_{\text{noisy}}=\Y+\mathbf{W}$ where $\mathbf{W} \in \mathcal{C}^{n_r \times T}$ is the random IID complex gaussian noise at the receiver.

The block fading model with block length $T$ is assumed.  In the model of channel of interest, in a block of time the fading coefficients change correlatively over time. The correlation matrix of fading coefficients between a pair of transmitter and receiver in a block of length $T$, $K_{\mathbf{H}}$, is of rank $Q<T$. 
Equivalently, all the fading coefficients between a pair of transmitter and receiver in a block of length $T$ are linear combinations of $Q$ statistically independent elements.  The case $Q=1$ corresponds to the block fading model in which the fading coefficients do not change in the block of length $T$. This is  the case studied in depth in \cite{Zheng}.

Define $\underline{\mathbf{h}}_{m,n}=[h_{m,n}\tm{1},\cdots,h_{m,n}\tm{T}]$  and $K_{\mathbf{H}}=\mathbb{E}[\underline{\mathbf{h}}^{\dagger}_{m,n} \underline{\mathbf{h}}_{m,n}] = \A^{\dagger} \A $.
The vector  $\underline{\mathbf{s}}_{m,n}=[s^1_{m,n},\cdots, s^Q_{m,n}]$ contains the $Q$ statistically independent elements whose linear combinations give the elements of vector $\underline{\mathbf{h}}_{m,n}$.
The matrix $\A\in \mathcal{C}^{Q \times T}$, known at both Tx and Rx, gives the linear equations which specify the fading coefficients $h_{m,n}\tm{t}$ from the independent, gaussian distributed numbers $\underline{\mathbf{s}}_{m,n}$ as follows:

\begin{equation}
	\label{eq:whitening} h_{m,n}\tm{t}=\sum_{q=1}^Q A^q\tm{t} s^q_{m,n}.
\end{equation}

Define $A^q\tm{t}$ to be the element in $q$th row and $t$th column of matrix $\A$ and $\underline{\A}^q$ to be the $q$th row of matrix $\A$

The noiseless received signal at the $n$th antenna would be:

\begin{equation}
	\label{eq:noiseless}
	\mathbf{y}_n=\sum_{m=1}^{n_t}\sum_{q=1}^{Q}s^q_{m,n} \underline{\A}^q \diag(\mathbf{x}_m),
	\end{equation}

where $\diag(\mathbf{x}_m)$ is a $T\times T$ matrix whose diagonal elements are the the elements of $\mathbf{x}_m$.

We observe that for all $n=1,\cdots,n_r$ the noiseless received signal at the $n$th antenna, $\mathbf{y}_n$, lives in the following subspace:

\begin{eqnarray}
\label{eq:Fdef}	
	\F_{\A}(\X) = \spn \{ \, \underline{\A}^q \diag(\mathbf{x}_m) & \mbox{ for  $q=1,\dots,Q$} \nonumber \\
	 &\mbox{ and $ m=1,\cdots,{n_{t}}$}\}.
\end{eqnarray}

$\F_{\A}(\X)$ is the output of a nonlinear transform which is a mapping from the linear subspace spanned by rows of matrix $\X$, $\Omega_{\X}$, to a higher dimensional subspace; parameterized by the matrix $\A$. 

\subsection{Reduced Row Echelon Form (RREF)}

\label{sec:RREF}
Since we are going to study and analyze a nonlinear mapping over the manifolds of different dimensions, we define a canonical representation of linear subspaces. Accordingly, the mapping over the manifolds can be represented as the mapping over the nontrivial elements of the canonical representations of the manifolds. 
Following the notation in \cite{Zheng}, each linear subspace of dimension $L$ in $\mathcal{C}^T$ can be represented as span of $L$ linearly independent vectors in the rows of a matrix  $\R \in \mathcal{C}^{L\times T}$. The same subspace is represented by choosing any non-singular matrix $\mathbf{C}\in \mathcal{C}^{L\times L}$ and constructing matrix $\B \in \mathcal{C}^{L\times T}$ such that $\R=\mathbf{C}\B$. The matrix $\mathbf{C}_{\R}$ can be chosen such that $\B\tm{1:L}=\I_{L}$, where $\B\tm{1:L}\in\mathcal{C}^{L \times L}$ is the submatrix of first $L$ columns of matrix $\B$ and $\I_{L}$ is the identity matrix of size $L$.



We call matrix $\B$ \emph{the reduced row echelon form} of this linear subspace of dimension $L$ in $\mathcal{C}^T$. Choosing $\mathbf{C}_{\R}=\R\tm{1:L}$ as the first $L$ columns of matrix $\R$, we construct $\B= \C^{-1}_{\R} \R$.

Note that it can be proved that for all the subspaces defined in this paper, the matrix $\C_{\R}$ is nonsingular with probability one. 


\section{$ n_r \geq Q$}
\subsection{Mapping over the Manifolds}

As shown in equation~\eqref{eq:noiseless} the noiseless received signals live in the subspace $\F_{\A}(\X)$ defined in equation~\eqref{eq:Fdef}. $\F_{\A}(\X)$ is the output of a nonlinear mapping from $\Omega_{\X}$, i.e., the linear subspace spanned by the rows of matrix $\X$,  to a higher dimensional linear subspace. If matrix $\X$ spans an $M$ dimensional linear subspace in $\mathcal{C}^T$ then $\F_{\A}(\X)$ spans an $MQ$ dimensional linear subspace in $\mathcal{C}^T$.

Looking at equation~\eqref{eq:noiseless}, it is clear that if $n_r \geq MQ$, then $\Omega_{\Y}=\F_{\A}(\X)$ with probability one. In this regime, $\F_{\A}(\X)$ can be recovered from the noiseless received signals with probability one. The linear subspace defined by $\F_{\A}(\X)$ has dimension $MQ$, but not all $MQ(T-MQ)$ degrees of freedom of linear subspaces of dimension $MQ$ in $\mathcal{C}^T$ are reachable by changing $\X$. We want to find the number of dimensions of $\F_{\A}(\X)$ which are only reachable by changing $\X$. To do so, we use the reduced row echelon form of the input and output signals of this nonlinear mapping over the manifolds. Then we prove that under given constraints, this mapping is bijective with probability one over a neighborhood of dimension $D$ in the input signal space and output signal space. 
Then, the theorem \ref{thm:DOF} helps us prove the achievability of $D$ degrees of freedom in the system.
 
Since $\F_{\A}(\X)$ is a mapping from $\Omega_{\X}$ to a higher dimensional linear subspace, the message is transmitted through the linear subspace spanning the rows of matrix $\X$. Let's assume that the matrix $\X$ is spanning an $M$ dimensional linear subspace. By fixing any $M$ columns of matrix $\widetilde{\X}=[\mathbf{x}_1^T,\cdots,\mathbf{x}_M^T]^T$ and choosing the other rows of matrix $\X$ to live in the span of $\widetilde{\X}$  this linear subspace is specified uniquely. This process is equivalent to using $M$ transmit antennas.
The nonlinear phase of the algorithm is preformed for $t=M Q+1,\cdots,M(Q+1)$. We choose the training signal to have the form $\widetilde{\X}\tm{MQ+1:M(Q+1)}=\I_{M}$. Thus,

\begin{equation}
	\label{eq:training} 
	x_m\tm{MQ+n}=\delta_{n-m} \quad \mbox{for} \quad 1\leq n,m \leq M.
\end{equation}

The matrix $\R \in \mathcal{C}^{MQ \times T}$ whose rows span $\F_{\A}(\X)$ is 

\begin{equation}
	\label{eq:R-MIMO}
	\R=\begin{bmatrix}
\A \diag(\mathbf{x}_1)\\
\A \diag(\mathbf{x}_2)\\
\vdots\\
\A \diag(\mathbf{x}_{M})\\
\end{bmatrix}
\end{equation}

Defining $\A\tm{1:MQ}\in \mathcal{C}^{Q\times M Q}$ to be the first $MQ$ columns of matrix $\A$ and  ${\mathbf{x}}_m\tm{1:M Q}$ to be a vector of the first $MQ$ elements of $\mathbf{x}_m$, the first $MQ$ columns of matrix $\R$ would look like:

\[\mathbf{C}_{\R}= \R\tm{1:MQ} = \begin{bmatrix}
\A\tm{1:{M}Q} \diag(\mathbf{x}_1\tm{1:{M}Q})\\
\A\tm{1:{M}Q} \diag(\mathbf{x}_2\tm{1:{M}Q})\\
\vdots\\
\A\tm{1:{M}Q} \diag(\mathbf{x}_{M}\tm{1:{M}Q})\\
\end{bmatrix}.\]

Define $\B$ as the canonical representation of $\F_{\A}(\X)$ in reduced row echelon form. Since $\B={\mathbf{C}_{\R}}^{-1}\R$ we know that the $t$th column of $\R$ and $t$th column of $\B$ would have the relation $\underline{\R}\tm{t}={\mathbf{C}_{\R}} \underline{\B}\tm{t}$ for all $t={M}Q+1,\cdots,T$.

The $t$th column of matrix $\R$ as defined in~\eqref{eq:R-MIMO} would be a vector of size $MQ\times 1$ which consists of stacking of $M$ vectors of size $Q \times 1$ each of which is a factor of $\underline{\A}\tm{t}$. The $m$th stack of $\R\tm{t}$  for $m=1,\cdots,{M}$ and $t > MQ$ would be $\underline{\A}\tm{t} x_m\tm{t}$, we would know

\begin{eqnarray}\underline{\A}\tm{t} x_m\tm{t}  & = & \A\tm{1:{M}Q}\diag( \mathbf{x}_m\tm{1:M Q})\underline{\B}\tm{t} \notag \\
	 & = & \A\tm{1:{M}Q}\diag(\underline{\B}\tm{t}) \mathbf{x}^T_m\tm{1:M Q}. \label{eq:MIMO1}
\end{eqnarray}


We use a two phase decoding algorithm. It uses the training signals in the first nonlinear phase to get information about $\widetilde{\X}\tm{1:M Q}$. In the linear phase of the decoding algorithm, the estimation of $\widetilde{\X}\tm{1:M Q}$ is used to recover the remaining transmitted signals, $\widetilde{\X}\tm{t}$ for $t>M(Q+1)$.

Define matrix $\mathbf{J}$ as

\[\mathbf{J}=\begin{bmatrix}
\A\tm{1:{M}Q} \diag(\underline{\B}\tm{M Q+1})\\
\A\tm{1:{M}Q} \diag(\underline{\B}\tm{M Q+2})\\
\vdots\\
\A\tm{1:{M}Q} \diag(\underline{\B}\tm{M(Q+1)})\\
\end{bmatrix}.\]


For each fixed  $m$ we could form the following equation which stacks the vectors in equation~\eqref{eq:MIMO1} for all $M Q<t \leq M(Q+1)$ in a matrix:

\begin{equation}
\begin{bmatrix}
\underline{\A}\tm{M Q+1} x_m\tm{M Q+1}\\
\underline{\A}\tm{M Q+2} x_m\tm{M Q+2} \\
\vdots\\
\underline{\A}\tm{M(Q+1))} x_m\tm{M(Q+1))}\\
\end{bmatrix}= \mathbf{J} \mathbf{x}^T_m\tm{1:MQ}.
\end{equation}

Considering the training signals as defined in equation~\eqref{eq:training}, we could form the following equation to be used in the nonlinear phase of the recovery algorithm:

\begin{equation}
	\label{eq:MIMO-nonlinear}
\widetilde{\X}^T= \mathbf{J}^{-1}
 \begin{bmatrix}
 \underline{\A}\tm{M Q+1} & 0 & & 0\\
  0 & \underline{\A}\tm{M Q+2} & & 0  \\
& & \ddots &\\
 0 & 0 &   &\underline{\A}\tm{MQ+M} \\
 \end{bmatrix}
\end{equation}

\subsection{Recovery Algorithm}

In this section, we summarize the recovery algorithm which decodes the transmitted message from the noiseless received signal with probability one. The transmitted message is embedded in the form of the $M$ dimensional linear subspace spanned by row of matrix $\widetilde{\X}$. The remaining rows of matrix $\X$ are designed to be in spans of rows of $\widetilde{\X}$ and do not carry information. 

\begin{enumerate}
	\item Construct estimation of ${\F_{\A}}(\X)$ as an $MQ$ dimensional linear subspace in $\mathcal{C}^T$. The best estimation of this subspace is the span of the low rank approximation of $\Omega_{\Y_{\text{noisy}}}$ using principal component analysis methods..
	\item Construct matrix ${\B}$ as the canonical representation of ${\F}_{\A}(\X)$ in reduced row echelon form.
	\item Use equation \eqref{eq:MIMO-nonlinear} to recover $\hat{\mathbf{x}}_m\tm{1:{M}Q}$ for $m=1,\cdots,M$.
	\item To recover ${x}_m\tm{t}$ for $m=1,\cdots,{M}$ and  ${M}(Q+1) < t \leq T$ use:
	\[ {x}_m\tm{t} = \underline{\A}^1\tm{1:{M}Q}\diag({\underline{\B}}\tm{t}) {\mathbf{x}}^T_m\tm{1:{M}Q}/A^1\tm{t}.\] 
\end{enumerate}

\subsection{Recovery Conditions}
\begin{enumerate}
	\item $M \leq n_t$: $M$ as the effective number of transmit antennas should be less than available number of transmit antennas
	\item $MQ \leq n_r$: This conditions guarantees the recovery of $\F_{\A}(\X)$ from the received signals with probability one.
	\item $M(Q+1)\leq T$: To perform the nonlinear phase of the decoding algorithm, the block length should be long enough to transmit the training signals.
	\item $\det(\mathbf{J})\neq 0$ with probability one: This is the required condition for the nonlinear mapping to be bijective with probability one. In theorem \ref{thm:detJ} a set of sufficient conditions for the mapping to be bijective is given.
	\end{enumerate}

To satisfy the first three conditions we choose \[M^*=\min(n_t,\lfloor{\frac{n_r}{Q}}\rfloor,\lfloor\frac{T}{Q+1}\rfloor).\]

\begin{theorem}
	\label{thm:detJ}
	If any choice of $Q$ columns of matrix $\A\tm{1:M(Q+1)}$ are linearly independent of each other, then with probability one $\det(\mathbf{J})\neq 0$.
	\end{theorem}
	
\begin{proof}
	We give the sketch of the proof here. It can be proved that $\det(\mathbf{J})$ is the ratio of two multivariate polynomial with variables in $\widetilde{\X}$. 
	
	Considering the fact that any set of zeros of a polynomial has measure zero as long as the polynomial is not identically equal to zero.
	 Thus, $\det(\mathbf{J})\neq 0$ with probability one \emph{iff} there is at least one term in the numerator with nonzero coefficient. 
	It can be proved that there is a term in the numerator of $\det(\mathbf{J})$ with variables $\left[ \prod_{m=1}^{M}\prod_{t=(m-1)Q+1}^{mQ}{x_m\tm{t}} \right]^{MQ-1}$ whose coefficient is
	
	\[\prod_{m=1}^{M}\left[{\det(\A\tm{(m-1)Q+1:mQ})}\right]^{M-1}{\det(\Lambda_m)},\]
	where $\Lambda_m$ is defined as:
	\[\Lambda_m =\A\tm{MQ+m,\{(m-1)Q+1:mQ\}\backslash \{(m-1)Q+m\}}.\]
	
	Determinant of $\Lambda_m$ is equal to the determinant of a matrix which is similar to $\A\tm{(m-1)Q+1:mQ}$, and its $m$th column is replaced with $\underline{\A}\tm{MQ+m}$.
	Clearly, a sufficient condition to have a polynomial which is not identically equal to zero is that all choices of $Q$ columns of matrix $\A\tm{1:M(Q+1)}$ are linearly independent. 
	
\end{proof}

\section{$n_r < Q $}
\label{sec:nrLQ}
In this regime the effective number of transmit antennas is one. This reduces the problem to SIMO case. The degrees of freedom in these systems is $\min(1-1/T,n_r(1-Q/T))$ as given in \cite{Riegler} and \cite{Riegler2}.

\section{Achievable Degrees of Freedom}
\begin{theorem}
	For the correlatively changing fading channels with correlation matrix $\mathbf{K}_{\mathbf{H}}={\A}^{\dagger} \A$, with  $n_r\geq Q$, define $M^*$ as
\[M^*=\min(n_t,\lfloor{\frac{n_r}{Q}}\rfloor,\lfloor\frac{T}{Q+1}\rfloor).\]

If there are $M^*(Q+1) $ columns in matrix $\A$ for which any choice of $Q$ columns are linearly independent then the following number of degrees of freedom per transmitted symbol is achievable:
\[\mbox{D}= M^*(1-M^*/T).\]

\end{theorem}

\section{Conclusion}

The formalized dimension counting argument is used to prove the achievability of degrees of freedom in MIMO correlatively changing fading channels.  This argument considers the mapping from the transmitted signals to the noiseless received signals. 

In the correlatively changing channel, this is a nonlinear mapping over manifolds of different dimensions. We studied this mapping and proposed a recovery algorithm as its inverse. We also proved that under generic condition the recovery algorithm is successful with probability one.

The proposed recovery algorithm exploits the nonlinear behavior of the mapping to achieve extra number of degrees of freedom. This nonlinear procedure is different by nature from the classic training approaches in noncoherent systems. We call the number of degrees of freedom achieved by this method the \emph{nonlinear degrees of freedom}.



%

\end{document}